\documentclass[a4paper,11pt]{article}

\usepackage{color}
\usepackage{latexsym,amssymb,amsmath}
\usepackage{amsthm}
\usepackage{graphicx}
\usepackage{enumerate}
\usepackage[noadjust]{cite}
\usepackage{url}

\usepackage{wrapfig}

\usepackage[top=2.54cm, bottom=2.54cm, left=2.54cm, right=2.54cm]{geometry}

\newtheorem{theorem}{Theorem}[section]
\newtheorem{proposition}[theorem]{Proposition}
\newtheorem{lemma}[theorem]{Lemma}
\newtheorem{corollary}[theorem]{Corollary}
\newtheorem{definition}[theorem]{Definition}

\newtheorem{observation}[theorem]{Observation}

\graphicspath{{figures/},{figures/boxes}}

\newcommand{\noproof}{~\hfill$\qed$}
\newcommand{\qedopt}{}

\newcommand{\ShoLong}[2]{#2} 

\title{Flip Distance Between Triangulations of a Simple Polygon is NP-Complete%
\thanks{Preliminary versions appeared as O.\@ Aichholzer, W.\@ Mulzer, and
A.\@ Pilz, \emph{Flip Distance Between Triangulations of a Simple 
Polygon is NP-Complete} in Proc.\@ 29th EuroCG, pp.~115--118, 2013, and
in Proc.\@ 21st ESA, pp.~13--24, 
2013~\cite{eurocg_version,esa_version}.
The final publication is available at Springer via
\protect\url{http://dx.doi.org/10.1007/s00454-015-9709-7}.}
}

\author{Oswin Aichholzer\thanks{Institute for Software Technology, Graz 
University of Technology, Austria.
Partially supported by the ESF EUROCORES
programme EuroGIGA - ComPoSe, Austrian Science Fund (FWF): I~648-N18.
\texttt{oaich@ist.tugraz.at.}}
\and Wolfgang Mulzer\thanks{
Institute of Computer Science, Freie Universit\"at Berlin, Germany.
Supported in part by DFG project MU/3501/1.
\texttt{mulzer@inf.fu-berlin.de.}}
\and Alexander Pilz\thanks{Recipient of a DOC-fellowship of the Austrian
Academy of Sciences at the Institute for Software Technology, Graz University
of Technology, Austria.
Part of this work has been done while this author was visiting the 
Departamento de Matem\'aticas, Universidad de Alcal\'a, Spain.
\texttt{apilz@ist.tugraz.at.}}
}

\begin{document}

\maketitle

\begin{abstract}
Let $T$ be a triangulation of a simple polygon.
A \emph{flip} in~$T$ is the operation of replacing one diagonal of~$T$
by a different one such that the resulting graph is again
a triangulation.  The \emph{flip distance} between two triangulations is the smallest
number of flips required to transform one triangulation into the
other.
For the special case of convex polygons,
the problem of determining the shortest flip distance between two triangulations is equivalent to determining the rotation distance between two binary trees, 
a central problem which is still open after over 25 years of intensive study.

We show that computing the flip distance between two
triangulations of a simple polygon is NP-hard.  This complements a recent
result that shows APX-hardness of determining the flip distance between two
triangulations of a planar point set.
\end{abstract}

\section{Introduction}
Let $P$ be a simple polygon in the plane, that is, a closed region bounded by a 
piece-wise linear, simple cycle.  A \emph{triangulation} of~$P$ is a geometric 
(straight-line) maximal outerplanar graph whose outer face is the complement 
of~$P$ and whose vertex set consists of the vertices of~$P$.  The edges
that are not on the outer face are called \emph{diagonals}.  Let $d$ be a diagonal
whose removal creates a convex quadrilateral. Replacing $d$ with the
other diagonal of the quadrilateral yields another triangulation of~$P$.
This operation is called a \emph{flip}.
The \emph{flip graph} of~$P$ is the abstract graph whose vertices are the
triangulations of~$P$ and in which two triangulations are adjacent if and only if
they differ by a single
flip.  We 
study the \emph{flip distance}, i.e., the minimum number
of flips required to transform a given source triangulation into a target
triangulation.

Edge flips became popular in the context of Delaunay
triangulations.  Lawson~\cite{lawson_connected} proved that any triangulation
of a planar $n$-point set can be transformed into any other by $O(n^2)$ flips. 
Hence, for every planar $n$-point set the flip graph 
is connected with diameter $O(n^2)$.
Later, Lawson showed that in fact every triangulation can be transformed to the Delaunay
triangulation by $O(n^2)$ flips that locally fix the Delaunay 
property~\cite{lawson_delaunay}. 
Hurtado, Noy, and Urrutia~\cite{hurtado_noy_urrutia} gave an example where the 
flip distance is $\Omega(n^2)$, 
and they showed that the same bounds hold for triangulations of simple polygons.
They also proved that if the polygon has $k$ reflex vertices, then
the flip graph has diameter $O(n + k^2)$. In particular, the flip graph of
any planar polygon has diameter $O(n^2)$.  Their result also generalizes the
well-known fact that the flip distance between any two triangulations of a
convex polygon is at most $2n - 10$, for $n > 12$. This was shown by Sleator, Tarjan, and Thurston~\cite{sleator} in their work on the flip distance in convex polygons.
The latter case is particularly 
interesting due to the correspondence between flips in triangulations of
convex polygons and rotations in binary trees:
The dual graph of such a triangulation is a binary tree, and a flip 
corresponds to a rotation in that tree; conversely, for every 
binary tree, a triangulation can be constructed.

We mention two further remarkable results on 
flip graphs for point sets.  Hanke, Ottmann, and
Schuierer~\cite{edge_flipping_distance} showed that the flip distance between
two triangulations is bounded by the number of crossings in their overlay. 
Eppstein~\cite{eppstein} gave a
polynomial-time algorithm for calculating a lower bound on the flip distance.
His bound is tight for point sets with no empty 5-gons; however, except for
small instances, such point sets are not in general position (i.e.,
they must contain collinear triples)~\cite{empty5gon}. 
A recent survey on flips is provided by Bose and Hurtado~\cite{survey}.

Recently, the problem of finding the flip distance between two triangulations of a
point set was shown to be NP-hard by Lubiw and
Pathak~\cite{lubiw} and, independently, by Pilz~\cite{point_set_hard}.
The latter proof was later improved to show APX-hardness of the problem.
A recent paper shows that the problem is fixed-parameter tractable~\cite{flip_distance_fpt}.
Here, we show that the corresponding problem remains NP-hard even for simple polygons.
This can be seen as a further step towards settling the complexity
of deciding the flip distance between triangulations of convex polygons or, equivalently, the rotation distance between binary trees.
This variant of the problem was probably first addressed by Culik and 
Wood~\cite{tree_similarity} in 1982 (showing a flip distance of $2n-6$) in the context of similarity measures between trees.

We now give the formal problem definition:
given a simple polygon~$P$, two triangulations 
$T_1$ and $T_2$ of~$P$, and an integer $l$, decide whether
$T_1$ can be transformed into $T_2$ by at most $l$ flips.
We call this decision problem \textsc{PolyFlip}.
To show NP-hardness, we give a polynomial-time reduction from
the problem \textsc{Rectilinear Steiner Arborescence} to \textsc{PolyFlip}.
\textsc{Rectilinear Steiner Arborescence} was shown to be NP-hard by Shi and Su~\cite{shi_su}.
In Section~\ref{sec:RSA}, we describe the problem in detail. 
We present the well-known \emph{double chain}
(used by Hurtado, Noy, and Urrutia~\cite{hurtado_noy_urrutia} for giving their lower bound), a major 
building block in our reduction, in Section~\ref{sec:double_chain}.
Finally, in Section~\ref{sec_reduction}, we describe our reduction
and prove that it is correct.

\section{The Rectilinear Steiner Arborescence Problem}
\label{sec:RSA}
Let $S$ be a set of~$N$ points in the plane whose coordinates are nonnegative integers. 
The points in~$S$ are called \emph{sinks}. 
A \emph{rectilinear tree} $A$ is a connected acyclic collection of horizontal 
and vertical line segments that intersect only at their endpoints.
The \emph{length} of~$A$ is the total length of all segments in~$A$ 
(cf.~\cite[p.~205]{hwang}).
The tree $A$ is a \emph{rectilinear Steiner tree} for $S$ if every sink in 
$S$ appears as an endpoint of a segment in~$A$. 
We call $A$ a \emph{rectilinear Steiner arborescence} (RSA) for $S$ if (i) $A$ is 
rooted at the origin; (ii) every leaf of~$A$ lies at a sink in~$S$; and (iii) for 
each $s = (x_s,y_s) \in S$, the length of the path in~$A$ from the origin to $s$ 
equals $x_s+y_s$, i.e., all edges in~$A$ point north or east, as seen from the origin~\cite{rao}.
In the problem 
\textsc{Rectilinear Steiner Arborescence}, we are given a set of sinks $S$ and an integer~$k$. The question
is whether there is an RSA for $S$ of length at most $k$. 
Shi and Su showed that \textsc{Rectilinear Steiner Arborescence} is strongly 
NP-complete; in particular, it 
remains NP-complete if $S$ is contained in an $n \times n$ grid, with $n$ 
polynomially bounded in~$N$, the number of
sinks~\cite{shi_su}.\footnote{Although a  polynomial-time algorithm was 
claimed~\cite{trubin}, it has later been shown to be incorrect~\cite{rao}.}

We will need the following important structural property of the RSA.
Let $A$ be an RSA for a set $S$ of sinks. 
Let $e$ be a vertical segment in~$A$ that does not contain a sink. 
Suppose there is a horizontal segment $f$ incident to the upper endpoint $a$ of~$e$.
Since $A$ is an arborescence, $a$ is the left endpoint of~$f$.
Suppose further that $a$ is not the lower endpoint of another vertical edge.
Take a copy $e'$ of~$e$ and translate it to the right until $e'$ hits a sink 
or another segment endpoint (this will certainly happen at the right
endpoint of~$f$); see \figurename~\ref{fig_arborescence_slide}.
The segments $e$ and $e'$ define a rectangle~$R$.
The upper and left side of~$R$ are completely covered by $e$ and (a part of) $f$.
Since $a$ has only two incident segments, every sink-root path in~$A$ that goes through 
$e$ or $f$ contains these two sides of~$R$, entering the boundary of~$R$ 
at the upper right corner $d$ and leaving it at the lower left corner~$b$.
We reroute every such path at $d$ to continue clockwise along the boundary of 
$R$ until it meets $A$ again (this certainly happens at~$b$), and we 
delete $e$ and the part of~$f$ on $R$.
In the resulting tree we subsequently remove all unnecessary 
segments (this happens if there are no more root-sink paths through~$b$) 
to obtain another RSA $A'$ for~$S$.
Then $A'$ is not longer than~$A$.
This operation is called \emph{sliding $e$ to the right}.
If similar conditions apply to a horizontal edge, we can \emph{slide it upwards}.
The \emph{Hanan grid} for a point set is the set of all vertical and horizontal 
lines through its points.
Through repeated segment slides in a shortest RSA, one
can obtain the following theorem.
\begin{theorem}[\cite{rao}]\label{thm_slide}
Let $S$ be a set of sinks. There is a minimum-length RSA $A$ for
$S$ such that
all segments of~$A$ are on the Hanan
grid for $S \cup \{(0,0)\}$.\noproof
\end{theorem}
\begin{figure}
\centering
\includegraphics[width=\textwidth]{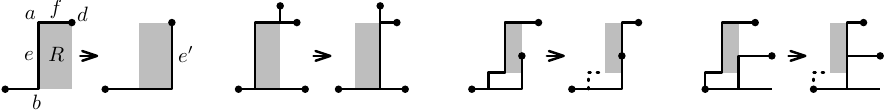}
\caption{The slide operation. The dots depict sinks; the 
rectangle~$R$ is drawn gray.
The dotted segments are deleted, since they do no longer lead to a sink.
}
\label{fig_arborescence_slide}
\end{figure}

We use a restricted version of \textsc{Rectilinear Steiner Arborescence}, 
called YRSA.
An instance $(S, k)$ of YRSA differs from an instance for \textsc{Rectilinear
  Steiner Arborescence}
in that we require that no two sinks in~$S$ have the 
same $y$-coordinate.
\ShoLong{The NP-hardness of YRSA follows by a simple perturbation argument; 
see the full version for all omitted proofs.}
{}

\begin{theorem}\label{thm_yrsa}
\textup{YRSA} is strongly \textup{NP}-complete.
\end{theorem}

\begin{proof}
Due to Theorem~\ref{thm_slide}, YRSA and 
\textsc{Rectilinear Steiner Arborescence} 
are in NP~\cite{shi_su}.
We now show how to transform an instance
$(S,k)$ of \textsc{Rectilinear Steiner Arborescence} to an
instance of YRSA. 
We may assume that $N = |S| \geq 3$, and we number the sinks as $S = \langle s_1, s_2, \dots, s_N \rangle$ in an arbitrary fashion.
For $i = 1, \dots, N$, let $(x_i, y_i)$ be the coordinates of~$s_i$ and define $s'_i := (x_iN^4, y_iN^4 + i)$.
We set $S' := \{s'_1, s'_2, \dots, s'_N\}$.
The $y$-coordinates of the sinks in~$S'$ are pairwise distinct.
We will show that there is an RSA for $S$ of length at most $k$
if and only if there is an RSA for $S'$ of length at most
$kN^4 + N^3$. 

Let $A$ be a rectilinear Steiner arborescence for $S$ of length at most $k$.
We scale $A$ by $N^4$ and draw a vertical segment from each leaf to the sink in $S'$
above it.
This gives an RSA for $S'$ of length at most $k N^4  + N^2 < k N^4 + N^3$.

Conversely, let $A'$ be an RSA for $S'$ of length at most $kN^4 + N^3$.
Due to Theorem~\ref{thm_slide}, we can assume that $A'$ is on the Hanan grid.
We round the $y$-coordinate of every segment endpoint in 
$A'$ down to the next multiple of $N^4$ (possibly removing segments of length 0).
The resulting drawing remains connected; every path to the origin 
remains monotone; and since the segments of $A'$ lie on the
Hanan grid of $S' \cup \{(0,0)\}$, no new cycles are introduced. 
Thus, the resulting drawing constitutes an arborescence $A''$ for 
the set $S''$ of sinks obtained by scaling $S$ by $N^4$.
Since $A'$ lies on the Hanan grid, it is a union of~$N$ paths, each with
at most $N$ vertical segments. 
The rounding operation increases the length of
each such vertical segment by at most $N$.
Thus, the total length of $A''$ is at most
$k N^4 + 2N^3$. By Theorem~\ref{thm_slide}
there exists an optimum arborescence $A^*$ for
$S''$ that lies on the Hanan grid. The
length of $A^*$ is a multiple of $N^4$, and thus
at most $kN^4$, since $2N^3 < N^4$ for $N \geq 3$.
 It follows that $S$ has an RSA of length at most $k$.

Therefore, $(S, k)$ is a yes-instance for \textsc{Rectilinear Steiner
Arborescence} if and only if $(S', kN^4 + N^3)$ is a yes-instance for
YRSA. Since $(S', kN^4 + N^3)$ can be computed in polynomial time 
from $(S, k)$, and since the coordinates in
$S'$ are polynomially bounded in the coordinates of~$S$, it follows that 
YRSA is strongly NP-complete.
\qedopt
\end{proof}

Due to Theorem~\ref{thm_slide}, we get the following technical corollary, which will be useful later.

\begin{corollary}\label{cor_blow_up}
YRSA remains strongly NP-complete even if the sinks have coordinates that are a multiple of a positive integer whose value is polynomial in~$N$.
\end{corollary}

\section{Double Chains}\label{sec:double_chain}
Our definitions (and illustrations) follow~\cite{point_set_hard}.  A \emph{double
chain} $D$ is a polygon
that consists of two chains, an \emph{upper chain} and a \emph{lower chain}.
There are $h$ vertices on each chain, $\langle u_0, \dots, u_{h-1} \rangle$ on the
upper chain and $\langle l_0, \dots, l_{h-1} \rangle$ on the lower chain,
both numbered from left to right, and $D$ is defined by $\langle l_0, \dots, l_{h-1}, u_{h-1}, \dots, u_0 \rangle$.
Any point on one chain sees every point on the other chain, and any quadrilateral formed by three vertices of one chain and one vertex of the other chain is non-convex;
see \figurename~\ref{fig_dc_triangulations_wedges}~(left).
We call the triangulation $T_u$ of~$D$ where $u_0$ has maximum degree the
\emph{upper extreme triangulation}; observe that this triangulation is unique.  
The triangulation $T_l$ of~$D$ where $l_0$ has maximum degree
is called the \emph{lower extreme triangulation}.  The two extreme triangulations
are used to show that the diameter of the flip graph is quadratic; see \figurename~\ref{fig_dc_triangulations_wedges}~(right).

\begin{theorem}[Hurtado, Noy, Urrutia~\cite{hurtado_noy_urrutia}]\label{thm_dc}
The flip distance between $T_u$ and $T_l$ is~\mbox{$(h-1)^2$}.\noproof
\end{theorem}

Through a slight modification of~$D$, we can make the flip distance between the upper and the lower extreme triangulation linear.
This will enable us in our reduction to impose a certain structure on short flip sequences.
To describe this modification, we first define the \emph{flip-kernel} of a double chain.

\begin{figure}
\centering
\includegraphics[width=\textwidth]{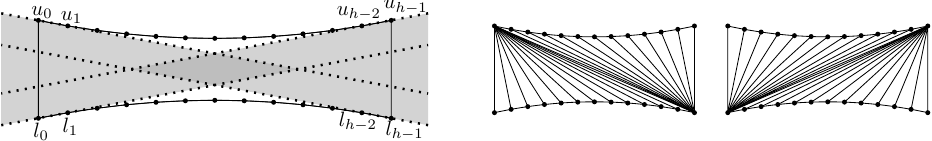}
\caption{
Left: The polygon and the hourglass (gray) of a double chain. The diamond-shaped
flip-kernel can be extended arbitrarily by flattening the chains.
Right: The upper extreme triangulation $T_u$ and the lower extreme triangulation $T_l$. 
}
\label{fig_dc_triangulations_wedges}
\end{figure}

Let $W_1$ be the wedge defined by the lines through $u_0 u_1$ and
$l_0 l_1$ whose interior contains no vertex of~$D$ but intersects the segment $u_0 l_0$.
Define $W_h$ analogously by the lines through $u_{h-1} u_{h-2}$ and $l_{h-1} l_{h-2}$.
We call $W := W_1 \cup W_h$ the \emph{hourglass of~$D$}.
The unbounded set $W \cup D$ is defined by four rays and the two chains.
The \emph{flip-kernel} of~$D$ is the intersection of the four closed half-planes below the lines through $u_0 u_1$ and $u_{h-2} u_{h-1}$ and above the lines through $l_0 l_1$ and $l_{h-2} l_{h-1}$.%
\footnote{The flip-kernel of~$D$ might not be completely inside the polygon~$D$.
This is in contrast to the ``visibility kernel'' of a polygon.}

\begin{figure}
\centering
\includegraphics[width=\textwidth]{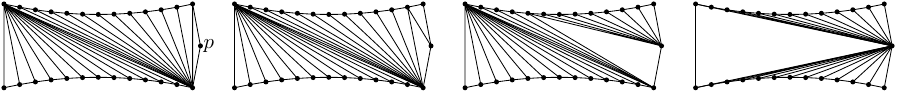}
\caption{The extra point~$p$ in the flip-kernel of~$D$ allows flipping one extreme triangulation
of~$Q$ to the other in $4h-4$ flips.}
\label{fig_dc_steiner}
\end{figure}

\begin{definition}\label{def_p_d_plus}
\sloppypar{
Let $D$ be a double chain and let $p$ be a point in the flip-kernel of~$D$ to the right
of the directed line $l_{h-1} u_{h-1}$.
The polygon given by the sequence $\langle l_0, \dots, l_{h-1}, p, u_{h-1}, \dots, u_0 \rangle$ is called a \emph{double chain extended by~$p$}.
The upper and the lower extreme triangulation of such a polygon contain the edge $u_{h-1} l_{h-1}$ as a diagonal and are otherwise defined in the same way as for~$D$.
}
\end{definition}

The flip distance between the two extreme
triangulations of $D$ extended by a point~$p$ is much
smaller than for $D$~\cite{problemas}.
\figurename~\ref{fig_dc_steiner} shows how to 
transform them into each other with $4h-4$ flips.
The next lemma shows that this is optimal, even for more general polygons.
The lemma is a slight generalization of a lemma by Lubiw and Pathak~\cite{lubiw} on double chains of constant size.

\begin{lemma}\label{lem_lower_bound}
\begin{sloppypar}
Suppose that $h \geq 5$ and
consider a polygon that contains $D$ and has $\langle l_0, \dots, l_{h-1} \rangle$ and $\langle u_{h-1},\dots, u_0\rangle$ as part of 
its boundary.
Let $T_1$ and $T_2$ be two triangulations that contain the upper
extreme triangulation and the lower extreme triangulation of~$D$ as a
sub-triangulation, respectively.  Then $T_1$ and $T_2$ have flip distance at
least $4h-4$.
\end{sloppypar}
\end{lemma}
\begin{proof}
We slightly generalize a proof by Lubiw and Pathak~\cite{lubiw} for
double chains of constant size.

Let $C_u$ be the upper chain and $C_l$ be the lower
chain of $D$.
The triangulation $T_1$ has $2(h-1)$ triangles
with an edge on $C_u$ or on $C_l$.
These triangles are called \emph{anchored},
and the vertex not incident to the edge on $C_u$ or on $C_l$ 
is called the \emph{apex}.
For each anchored triangle with an edge on $C_u$,
the apex must move from $l_{h-1}$ to~$l_0$, and similarly for $C_l$.
We distinguish three
types of flips depending on whether the convex quadrilateral whose diagonal
is flipped has (1) four; (2) three; or
(3) at most two vertices on $D$.
A flip of type (1) moves the apex
of two anchored triangles by one; a flip of type (2) moves the apex of one 
anchored triangle from
$D$ to a point outside $D$ or back again; and a flip of type (3)
does not move any apex of an anchored triangle along $D$ or between a vertex 
of~$D$ and a vertex not in~$D$.

We say that an anchored triangle is of type (1) if its apex is moved
only by flips of type (1). It is of type (2) if its apex
is moved by at least one flip of type (2). Every anchored triangle
is either of type (1) or of type (2).
A type (1) triangle must be involved
in at least $h-1$ flips of type (1), and each of these flips
can affect at most one other type (1) triangle.
A type (2) triangle must be involved in at least $2$
flips of type (2), and each of these flips can affect
no other anchored triangle.
Thus, if we have $m_1$ type (1) triangles and $m_2$ type (2)
triangles, we need at least
$(h-1)m_1/2 + 2m_2$ flips.
For $h \geq 5$, we have
$(h-1)m_1/2 + 2m_2 \geq 2(m_1 + m_2) = 4h -4$,
as claimed.
\qedopt
\end{proof}

The following result can be seen as a special case of~\cite[Proposition~1]{point_set_hard}.

\begin{lemma}\label{lem_empty_wedge}
\begin{sloppypar}
Consider a polygon that contains~$D$ and has $\langle u_{h-1}, \dots, u_0, l_0, \dots, l_{h-1} \rangle$ as part of 
its boundary. Let $T_1$ and $T_2$ be two triangulations that contain the upper and 
the lower extreme triangulation of~$D$ as a sub-triangulation, respectively.
Let $\sigma$ be a flip sequence from $T_1$ to $T_2$ such that there is no triangulation in~$\sigma$ containing a triangle with one vertex at the upper chain, the other vertex at the lower chain, and the third vertex at a point in the interior of the hourglass of~$D$.
Then $|\sigma| \geq (h-1)^2$.
\end{sloppypar}
\end{lemma}
\begin{proof}
Our reasoning is similar to the proof of Lemma~\ref{lem_lower_bound}, 
see also~\cite{lubiw}. As before, let $C_u$ and $C_l$
be the upper and lower chain of~$D$, and call a
triangle with an edge on $C_u$ or on $C_l$ \emph{anchored},
the third vertex being the \emph{apex}.
Any triangulation of the given polygon has $2(h-1)$ anchored triangles.

We will argue that for each triangulation of~$\sigma$ there 
exists a line~$\ell$ that separates $C_u$ from $C_l$
and that intersects all anchored triangles.
This is clear if the apices of all anchored triangles
lie on the other chain or outside the hourglass.
Now consider a triangulation of the sequence~$\sigma$ where 
at least one anchored triangle has its apex at a 
vertex $v$ inside the hourglass.
Let $r$ be a ray that starts at a point on $u_0 l_0$ and 
passes through~$v$ such that the supporting line of~$r$ separates 
$C_u$ from $C_l$ (such a ray must exist 
since $v$ is inside the hourglass).
Then $r$ intersects at least one triangle that 
is not anchored, because the triangle whose interior is intersected 
by~$r$ before reaching $v$ cannot be anchored.
Let $\Delta$ be the first non-anchored triangle whose 
interior is intersected by~$r$.
Then $\Delta$ has one vertex on $C_u$
and one vertex on $C_l$. By assumption, the 
third vertex of~$\Delta$  cannot be inside the hourglass,
so it must lie outside.
This means that one of the vertices 
of~$\Delta$ has to be either $u_{h-1}$ or~$l_{h-1}$.
This  implies that either all anchored triangles at $C_u$ or
or all anchored triangles $C_l$, 
respectively, have their apex at the opposite chain.
Thus, also for this triangulation there exists a line~$\ell$ 
that separates $C_u$ from $C_l$ and that 
intersects all anchored triangles. Observe that every such line
intersects the anchored triangles in the same order.

Now we proceed similarly as in the proof 
of Hurtado, Noy, and Urrutia~\cite{hurtado_noy_urrutia}: 
we observe that an anchored triangle at $C_u$  and an anchored triangle 
at $C_l$ can change their relative position 
along~$\ell$ only if they have an edge in common and this edge is flipped.
This results in an overall number of $(h-1)^2$ flips.
\qedopt
\end{proof}

\section{The Reduction}\label{sec_reduction}
We reduce YRSA to \textsc{PolyFlip}.
Let $S$ be a set of~$N$ sinks.
By Corollary~\ref{cor_blow_up}, we can assume that the coordinates of sinks of~$S$ are multiples of a factor~$\beta = 2N$ in $\{0, \dots , \beta n \}$.
Further, we can restrict ourselves to YRSA instances of the form $(S, \beta k)$.
Thus, we imagine that the sinks are embedded on a $\beta n \times \beta n$ grid.
The reasons for the choice of~$\beta$ will become clear below.

We construct a polygon~$P$ and two triangulations $T_1$, $T_2$ in~$P$ such that a shortest flip sequence from $T_1$ to $T_2$ corresponds to a shortest RSA for $S$.
To this end, we will describe how to interpret any triangulation of~$P$ as a \emph{chain path}, a path in the integer grid that starts at the root and uses only edges that go north or east.
It will turn out that flips in~$P$ essentially correspond to moving the endpoint of the chain path along the grid.
We choose $P$, $T_1$, and $T_2$ in such a way that a shortest flip sequence between $T_1$ and $T_2$ moves the endpoint of the chain path according to an Eulerian traversal of a shortest RSA for~$S$.
To force the chain path to visit all sinks, we use the observations from Section~\ref{sec:double_chain}: the polygon $P$ contains a double chain for each sink, so that only for certain triangulations of~$P$ it is possible to flip the double chain quickly.
These triangulations will be exactly the triangulations that correspond to the chain path visiting
the appropriate sink.
To force the sinks to be visited, we, with foresight, fix the number of points in each of the two chains of a double chain representing a sink to $d = nN$ (recall that $n$ is polynomial in~$N$).

\subsection{The Construction}
\label{sec_construction}
\sloppypar{
We take a double chain~$D$ with $\beta n + 2$ vertices on each chain such that the flip-kernel of~$D$ extends to the right of $l_{\beta n+1} u_{\beta n+1}$.
We add a point $z$ to that part of the flip-kernel, and we let $Q$ be the polygon defined by $\langle l_0, \dots, l_{\beta n+1}, z, u_{\beta n+1}, \dots, u_0 \rangle$, i.e., a double chain extended by~$z$ (recall Definition~\ref{def_p_d_plus}).
}
Next, we add double chains to $Q$ in order to encode the sinks in $S$.
For each sink $s = (x_s,y_s)$, we remove the edge $l_{y_s} l_{y_s + 1}$, and we replace it by a (rotated) double chain $D_s$ with $d$ vertices on each chain, such that $l_{y_s}$ and $l_{y_s+1}$ become the last point on the lower and the upper chain of $D_s$, respectively.
We orient $D_s$ in such a way that $u_{x_s}$ is the only point inside the hourglass of $D_s$ and so that $u_{x_s}$ lies in the flip-kernel of $D_s$; see \figurename~\ref{fig_installing_sites}.
We refer to the added double chains as the \emph{sink gadgets}, and
we call the resulting polygon $P$.
Since the $y$-coordinates in~$S$ are pairwise distinct, there is at most one sink gadget per edge of the lower chain of~$Q$.
Since $\beta \geq 2$, no two sink gadgets are placed on neighboring edges of~$Q$, and can be constructed such that they do not overlap.
Hence, $P$ is a simple polygon.
The precise placement of the sink gadgets is flexible, so given an appropriate embedding of~$D$, we can make all coordinates integers whose value is polynomial in the input size;
see Appendix~\ref{apx_coordinates} for details.

\begin{figure}
\centering
\includegraphics[width=\textwidth]{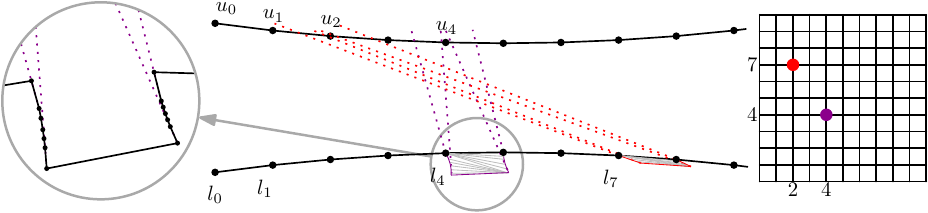}
\caption{
The sink gadget for a sink $(x_s,y_s)$ is obtained by replacing the edge
$l_{y_s} l_{y_s + 1}$ by a double chain with $d$ vertices on
each chain. The double chain is oriented such that $u_{x_s}$ is the
only point inside its hourglass and its flip-kernel. %
}
\label{fig_installing_sites}
\end{figure}

Next, we describe the source and target triangulation for $P$.
In the source triangulation $T_1$, the interior of~$Q$ is triangulated 
such that all edges are incident to~$z$. 
The sink gadgets are all triangulated with the upper extreme triangulation.
The target triangulation $T_2$ is similar, but now the sink gadgets 
are all triangulated with the lower extreme triangulation.

To get from $T_1$ to $T_2$, we must go from one extreme
triangulation to the other for each sink gadget $D_s$. 
By Lemma~\ref{lem_empty_wedge},
this requires $(d-1)^2$ flips, unless the flip sequence creates a
triangle that allows us to use the vertex in the flip-kernel of
$D_s$. In this case, we say that the flip sequence \emph{visits}
the sink $s$.
The main idea is that, since the value chosen for~$d$ is large, a shortest flip sequence must visit all sinks, 
and we  will show that this induces an RSA for~$S$ of comparable length.
Conversely, we will show how to derive a flip sequence from an 
RSA.
The precise statement is given in the following theorem.

\begin{theorem}\label{thm:YRSA<->FlipDist}
Let $N \geq 3$, and
set $\beta = 2N$.
Let $S$ be a set of $N$ sinks such that
the coordinates of
the sinks are multiples of $\beta$ in $\{0, \dots, \beta n\}$, where $n$ is polynomially bounded in~$N$.
Set $d = nN$ and let $P$ be the simple polygon and $T_1$ and $T_2$
the two triangulations of~$P$ as described above.
Then for any $k \geq 1$, the flip distance between $T_1$ and
$T_2$ w.r.t.~$P$ is at most $2\beta k + (4d-2)N$ if and only
if
$S$ has an RSA of length at most $\beta k$.
\end{theorem}

We will prove Theorem~\ref{thm:YRSA<->FlipDist} in the following
sections. But first, let us show how to use it for our NP-completeness
result.
\begin{theorem}
\textup{\textsc{PolyFlip}} is \textup{NP}-complete.
\end{theorem}
\begin{proof}
As mentioned in the introduction, the flip distance in polygons is
polynomially bounded, so \textsc{PolyFlip} is in NP. 
We reduce from YRSA.
Let $(S,\beta k)$ be an instance of YRSA as above. 
We construct $P$ and $T_1$, $T_2$ as described above.
This takes polynomial time (see Appendix~\ref{apx_coordinates} for details on the coordinate representation).
By Theorem~\ref{thm:YRSA<->FlipDist}, 
there exists an RSA for $S$ of length at most~$\beta k$ if and only if there exists a flip sequence between $T_1$ and $T_2$ of length at most
$2\beta k + (4d-2)N$. 
\qedopt
\end{proof}

\subsection{Chain Paths}
\label{sec:triangulation_structure}

Now we introduce the \emph{chain path}, our main tool to
establish a correspondence between flip sequences and RSAs.
Let $T$ be a triangulation of~$Q$ (i.e., the polygon $P$ 
without the sink gadgets, cf. Section~\ref{sec_construction}).
A \emph{chain edge} is an edge of~$T$ between the upper and the lower 
chain of~$Q$. 
A \emph{chain triangle} is a triangle of~$T$ that contains two chain edges.
Let $e_1, \dots, e_m$ be the chain edges, sorted from left to 
right according to their intersection with a line
that separates the upper from the lower chain.
For $i = 1, \dots, m$, write $e_i = (u_v, l_w)$ and set $c_i = (v,w)$.
In particular, $c_1 = (0,0)$.
Since $T$ is a triangulation, any two consecutive edges $e_i$, $e_{i+1}$ share one endpoint, while the other endpoints are adjacent on the corresponding chain.
Thus, $c_{i+1}$ dominates $c_i$ and $\|c_{i+1}-c_i\|_1 = 1$.
It follows that $c_1c_2\dots c_m$ is an $x$- and $y$-monotone path, 
beginning at the root.
It is called the \emph{chain path} for $T$.
Each vertex of the chain path corresponds to a chain edge, and each 
edge of the chain path corresponds to a chain triangle. Conversely,
every chain path induces a triangulation $T$ of~$Q$;
see \figurename~\ref{fig_chain_path}.
In the following, we let $b$ denote the upper right endpoint of the 
chain path.
\begin{figure}
\centering
\includegraphics{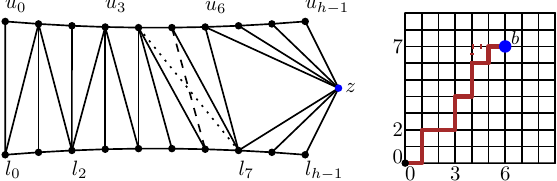}
\caption{A triangulation of~$Q$ and its
chain path. Flipping edges to and from $z$ moves the endpoint $b$ along
the grid.
A flip between chain triangles changes a bend. 
}
\label{fig_chain_path}
\end{figure}
\ShoLong{}{%
We now investigate how flipping edges in~$T$ affects the chain path.

\begin{observation}\label{obs_chain_path_monotone}
Suppose we flip an edge that is incident to~$z$.
Then the chain path is extended by moving $b$ north or east.\noproof
\end{observation}

\begin{observation}
Suppose that $T$ contains at least one chain triangle.
When we flip the rightmost chain edge, we shorten the chain path at $b$.\noproof
\end{observation}

Finally, we can flip an edge between two chain triangles.
This operation is called a \emph{chain~flip}.

\begin{observation}\label{lem_bend_chain_flip}
A chain flip changes a bend from east to north to a bend from north to east, or vice versa.
\end{observation}
\begin{proof}
If a chain edge $u_il_j$ is incident to two chain triangles and is flippable, then the two triangles must be of the form $u_i u_{i-1} l_j$ and
$l_j l_{j+1} u_i$, or $u_{i+1} u_i l_j$ and $l_{j-1} l_j u_i$.
Thus, flipping $u_il_j$ corresponds exactly to the claimed change in the chain path.
\qedopt
\end{proof}
}%
\ShoLong{}{%
\begin{corollary}
A chain flip does not change the length of the chain path.\noproof
\end{corollary}
}
\ShoLong{%
The next lemma describes how the chain path is affected by flips;
see Fig.~\ref{fig_chain_path}.
}
{%
We summarize the results of this section in the following lemma:
}
\begin{lemma}\label{lem:structure}
Any triangulation $T$ of~$Q$ uniquely determines a chain path, and vice versa.
A flip in~$T$ corresponds to one of the following operations on the chain path:
(i) move the endpoint $b$ north or east;
(ii) shorten the path at $b$;
(iii) change an east-north bend to a north-east bend, or vice versa.\noproof
\end{lemma}

\subsection{From an RSA to a Short Flip Sequence}
Using the notion of a chain path, we now prove the
``if'' direction of Theorem~\ref{thm:YRSA<->FlipDist}.

\begin{lemma}\label{lem:YRTSP->flip_tour}
Let $k \geq 1$ and $A$ an RSA for $S$ of length $\beta k$.
Then the flip distance between $T_1$ and $T_2$ w.r.t.~$P$ is at 
most $2\beta k + (4d-2)N$.
\end{lemma}
\begin{proof}
The triangulations $T_1$ and $T_2$ both contain a triangulation of~$Q$ 
whose chain path has its endpoint $b$  at the root.
We use Lemma~\ref{lem:structure} to generate flips inside $Q$ so
that $b$ traverses $A$ in a depth-first manner.
This needs $2\beta k$ flips.

Each time $b$ reaches a sink $s$, we move $b$ north.
This creates a chain triangle that allows the edges in
the sink gadget $D_s$ to be flipped to the auxiliary vertex
in the flip-kernel of $D_s$. The triangulation of $D_s$ can then be changed
with $4d-4$ flips; see Lemma~\ref{lem_lower_bound}.
Next, we move $b$ back south and continue the traversal.
Moving $b$ at~$s$ needs two additional flips, so we take $4d-2$ flips per sink,
for a total of $2\beta k + (4d-2)N$ flips.
\qedopt
\end{proof}

\subsection{From a Short Flip Sequence to an RSA}
Finally, we consider the ``only if'' direction in Theorem~\ref{thm:YRSA<->FlipDist}.
Let $\tau$ be a flip sequence on~$Q$.
We say that $\tau$ \emph{visits} a sink $s = (x_s,y_s)$ if $\tau$ has at least one triangulation that contains the chain 
triangle $u_{x_s} l_{y_s} l_{y_s + 1}$.
We call $\tau$ a \emph{flip traversal} for~$S$ if (i) $\tau$ begins and ends in the triangulation whose corresponding chain path has its endpoint~$b$ at the root and (ii) $\tau$ visits every sink in~$S$.
The following lemma shows that every short flip sequence~$\sigma$ in~$P$ can be mapped to a flip traversal (where with ``short'', we mean $|\sigma| < (d-1)^2$).

\begin{lemma}\label{lem:flip_traversal}
Let $\sigma$ be a flip sequence from $T_1$ to $T_2$ w.r.t.~$P$ 
with $|\sigma| < (d-1)^2$.
Then there is a flip traversal $\tau$ for $S$ 
with $|\tau| \leq |\sigma|-(4d-4)N$.
\end{lemma}
\begin{proof}
We show how to obtain a flip traversal $\tau$ for
$S$ from~$\sigma$.
Let $T$ be a triangulation of~$P$.
A triangle of~$T$ is an \emph{inner triangle} if all its sides are 
diagonals. 
It is an \emph{ear} if two of its sides are polygon edges.
By construction, every inner triangle of~$T$ must have
(i) one vertex incident to $z$ (the rightmost vertex of~$Q$), 
or (ii) two vertices incident to a sink gadget (or both).
There can be only one triangle of type (ii) per sink gadget.
The weak (graph theoretic) dual of~$T$ is a tree 
in which ears correspond to leaves and inner triangles have degree $3$.

For a sink~$s = (x_s,y_s)$, let $D_s$ be the corresponding sink gadget.
It lies between the vertices $l_{y_s}$ and $l_{y_s + 1}$ and has exactly $u_{x_s}$ in its flip kernel.
For brevity, we will write $l_s$ for $l_{y_s}$, $l'_{s}$ for $l_{y_s+1}$, and $u_s$ for $u_{x_s}$.
We define a triangle $\Delta_s$  for $D_s$.
Consider the bottommost edge~$e$ of~$D_s$, and let $\Delta$ be the triangle of~$T$ that is incident to $e$.
By construction, $\Delta$ is either an ear of~$T$, or it is the triangle defined by $e$ and $u_s$.
In the latter case, we set $\Delta_s = \Delta$.
In the former case, we claim that $T$ has an inner triangle $\Delta'$ with two vertices on $D_s$:
follow the path from $\Delta$ in the weak dual of~$T$;
while the path does not encounter an inner triangle, the next triangle must have an edge of $D_s$ as a side.
There is only a limited number of such edges, so eventually we must meet an inner triangle $\Delta'$.
We then set $\Delta_s = \Delta'$; see \figurename~\ref{fig_delta_s}.
Note that $\Delta_s$ might be $l_s l_s' u_s$.

\begin{figure}
\centering
\includegraphics{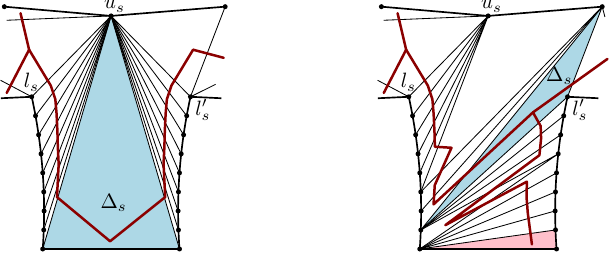}
\caption{Triangulations of $D_s$ in~$P$ with $\Delta_s = \Delta$ (left),
and with $\Delta$ being an ear (red) and $\Delta_s$ an inner triangle (right).
The fat tree indicates the dual.
}
\label{fig_delta_s}
\end{figure}

For each sink $s$, let the polygon $Q_s$ consist of $D_s$ extended by the vertex~$u_s$ (cf.~Definition~\ref{def_p_d_plus}).
Let $T$ be a triangulation of~$P$.
We show how to map $T$ to a triangulation $T_Q$ of~$Q$ and to triangulations $T_s$ of $Q_s$, for each~$s$.

We first describe $T_Q$.
It contains every triangle of~$T$ with all three vertices in~$Q$.
For each triangle $\Delta$ in~$T$ with two vertices on $Q$ and one vertex on the left chain of a sink gadget $D_s$, we replace the vertex on $D_s$ by~$l_s$.
Similarly, if the third vertex of $\Delta$ is on the right chain of $D_s$, we replace it by $l_s'$.
For every sink $s$, the triangle $\Delta_s$ has one vertex at a point $u_i$ of the upper chain.
In $T_Q$, we replace $\Delta_s$ by the triangle $l_s l_s' u_i$.
No two triangles in $T_Q$ overlap, and they cover all of~$Q$.
Thus, $T_Q$ is indeed a triangulation of~$Q$.

Now we describe how to obtain $T_s$, for a sink $s \in S$.
Each triangle of~$T$ with all vertices on $Q_s$ is also in~$T_s$. 
Each triangle with two vertices on $D_s$ and one vertex not in 
$Q_s$ is replaced in $T_s$ by a triangle whose 
third vertex is moved to
$u_s$ in $T_s$ (note that this includes~$\Delta_s$);
see \figurename~\ref{fig_local_triangulations}. Again, all triangles
cover $Q_s$ and no two triangles overlap.

\begin{figure}
\centering
\includegraphics[width=\textwidth]{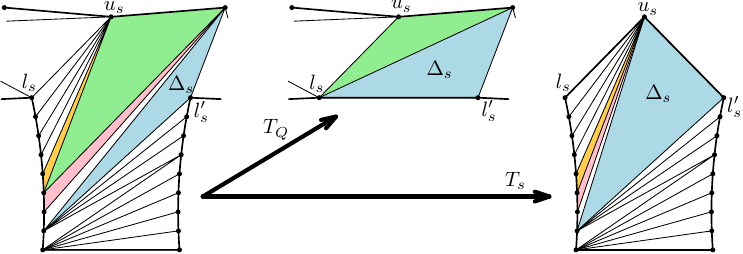}
\caption{Obtaining $T_Q$ and $T_s$ from $T$.}
\label{fig_local_triangulations}
\end{figure}

Finally, we show that a flip in~$T$ corresponds to at most one flip either in $T_Q$ or in precisely one $T_s$ for some sink~$s$.
We do this by considering all the possibilities for two triangles that share a common flippable edge.
By construction, no two triangles that are mapped to two different triangulations $T_s$ and $T_t$ for sinks $s \neq t \in S$ can share an edge.

\textbf{Case 1.}  We flip an edge between two triangles that are either both mapped to $T_Q$ or to $T_s$ and are different from $\Delta_s$.
This flip clearly happens in at most one triangulation.

\textbf{Case 2.} We flip an edge between a triangle~$\Delta_1$ that is mapped to $T_s$ and a triangle~$\Delta_2$ that is mapped to $T_Q$, such that both $\Delta_1$ and $\Delta_2$ are different from~$\Delta_s$.
This results in a triangle $\Delta_1'$ that is incident to the same edge of $Q_s$ as $\Delta_1$, and a triangle~$\Delta_2'$ having the same vertices of~$Q$
as~$\Delta_2$.
Since the apex of~$\Delta_1$ is a vertex of the upper chain or~$z$
(otherwise, it would not share an edge with $\Delta_2$), it is
mapped to $u_s$, as is the apex of~$\Delta_1'$.  Also, the apex of~$\Delta_2'$ is on
the same chain of~$D_s$ as the one of~$\Delta_2$.  Hence, the flip affects neither $T_Q$ nor~$T_s$.

\textbf{Case 3.}
We flip the edge between a triangle $\Delta_2$ mapped to $T_Q$ and
$\Delta_s$.
By construction, this can only happen if $\Delta_s$ is an inner triangle.
The flip affects only $T_Q$, because the new inner triangle
$\Delta_s'$ is mapped to the same triangle in $T_s$ as $\Delta_s$, since both
apexes are moved to $u_s$.

\textbf{Case 4.}
We flip the edge between a
triangle $\Delta$ of $T_s$ and $\Delta_s$. Similar to Case~3, this
affects only $T_s$, because the new
triangle~$\Delta_s'$ is mapped to the same triangle in $T_Q$ as~$\Delta_s$, 
since the two corners are always
mapped to $l_s$ and $l_s'$.

Thus, $\sigma$ induces a flip sequence
$\tau$ in~$Q$ and flip sequences $\sigma_s$ in each $Q_s$
so that $|\tau| + \sum_{s \in S} |\sigma_s| \leq |\sigma|$.
Furthermore, each flip sequence $\sigma_s$ transforms $Q_s$
from one extreme triangulation to the other.
Since $|\sigma| < (d-1)^2$, Lemma~\ref{lem_empty_wedge} tells us that the 
triangulations~$T_s$ have to be transformed so that $\Delta_s$ has a vertex at $u_s$ at some point. Moreover, by Lemma~\ref{lem_lower_bound}, we have
$|\sigma_s| \geq 4d-4$ for each $s \in S$.
Thus, $\tau$ is a flip traversal,
and $|\tau| \leq |\sigma| - N(4d-4)$, as claimed.
\qedopt
\end{proof}

In order to obtain a static RSA from a changing flip traversal,
we use the notion of a \emph{trace}.
A \emph{trace} is a domain on the %
grid.
It consists  of \emph{edges} and \emph{boxes}: an edge is a line segment of 
length $1$ whose endpoints have positive integer coordinates; a box is a 
square of side length~$1$ whose corners have positive integer coordinates.
Similar to arborescences, we require that a trace $R$ (i) is (topologically) 
connected; (ii) contains the root $(0,0)$; and (iii) from every grid point 
contained in~$R$ there exists an $x$- and $y$-monotone path to the root 
that lies completely in~$R$.
We say $R$ is a \emph{covering trace} for~$S$ (or, $R$ \emph{covers} $S$) if every sink in~$S$ is covered by~$R$ (i.e., incident to a box or an edge in~$R$).

Let $\tau$ be a flip traversal as in Lemma~\ref{lem:flip_traversal}.
By Lemma~\ref{lem:structure}, each triangulation in~$\tau$ 
corresponds to a chain path.
This gives a covering trace $R$ for $S$ in the following way.
For every flip in $\tau$ that extends 
the chain path, we add the corresponding edge to $R$.
For every flip in $\tau$ that changes a bend, we add the corresponding 
box to $R$. Afterwards, we remove from $R$ all edges that coincide 
with a side of a box in~$R$.
Clearly, $R$ is (topologically) connected. Since $\tau$ is a 
flip traversal for $S$, every sink
is covered by~$R$.
Note that every grid point $p$ in~$R$ is connected to the root by an 
$x$- and $y$-monotone path on $R$, since at some point $p$ belonged
to a chain path in~$\tau$.
Hence, $R$ is indeed a trace, the unique \emph{trace of~$\tau$}.
Note that not only a flip traversal but any flip sequence starting with a zero-length chain path defines a trace in this way.

Next, we define the \emph{cost} of a trace $R$, $\text{cost}(R)$, so that 
if $R$ is the trace of a flip traversal~$\tau$, then $\text{cost}(R)$ 
gives a lower bound on $|\tau|$. 
An edge has cost $2$.
Let $B$ be a box in~$R$.
A \emph{boundary side} of~$B$ is a side that is not part of another box.
The cost of~$B$ is~$1$ plus the number of boundary sides of~$B$.
Then, $\text{cost}(R)$ is the total cost over all boxes and edges in~$R$.
For example, the cost of a tree is twice the number of its edges, and the cost of 
a rectangle is its area plus its perimeter.
An edge can be interpreted as a degenerated box, having two boundary 
sides and no interior.

\begin{proposition}\label{prp_cost}
Let $\tau$ be a flip traversal and
$R$ the trace of  $\tau$. Then $\text{cost}(R) \leq |\tau|$.
\end{proposition}
\begin{proof}
Let $\varsigma_i$ be the sequence of the first $i$ triangulations of~$\tau$, $R_i$ the trace defined by~$\varsigma_i$, and let $\kappa_i$ be the length of the chain path for the $i$th triangulation.
We will show by induction on $i$ that $\text{cost}(R_i) \leq i + \kappa_i$, for $i = 1, \dots, |\tau|$.
Since $\varsigma_{|\tau|} = \tau$, $R_{|\tau|} = R$, and $\kappa_{|\tau|} = 0$, this gives the desired result.

After the first flip, $R_1$ is an edge (so $\text{cost}(R_1) = 2$), and $\kappa_1 = 1$, which fulfills the invariant.
Consider the $i$th flip.
If the flip extends the chain path, the cost of the trace increases by at most~2, and the length of the chain path increases by~1, fulfilling the invariant.
If the flip contracts the chain path, the trace does not change, but the length of the chain path is decreased by~1, again fulfilling the invariant.
We are therefore left with the case where the flip is a chain flip.
We have $\kappa_{i-1} = \kappa_i$, so we have to show that $\text{cost}(R_i) \leq \text{cost}(R_{i-1}) + 1$.
We may assume that the flip adds a box~$B$ to $R_{i-1}$ (otherwise,
the cost of the trace remains unchanged).
Consider the intersection of the boundary of~$B$ with the one of~$R_{i-1}$.
This intersection contains at least two elements, as the chain path is part of~$R_{i-1}$.
An edge in the intersection becomes a boundary side in~$R_i$, reducing the cost by~$1$.
A boundary side in the intersection vanishes in~$R_i$, also reducing the cost by~$1$.
Thus, adding $B$ creates a box and at most two boundary sides, causing a cost 
of at most~3, but it simultaneously reduces the cost by at least~$2$.
See the examples in \figurename~\ref{fig_boundary_cost}.
The overall cost increases at most by~$1$, and the invariant is maintained.
\qedopt
\end{proof}

\begin{figure}[ht]
\centering
\includegraphics[scale=1.2]{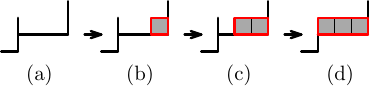}
\caption{Examples of how boundary sides (red) are added to a trace.
To a trace of cost 16 (a) a box (gray) is added (b), which transforms 
two edges in boundary sides and adds two boundary sides, resulting in 
an overall cost of 17.
The next box removes one boundary side and one edge and adds three 
boundary sides (c), the cost becomes 18.
A box might also remove more than two elements (d), reducing the overall 
cost to 17.
}
\label{fig_boundary_cost}
\end{figure}

Now we relate the length of an RSA for $S$ to the cost of a covering trace for $S$, and thus to the length of a flip traversal.
Since each sink is connected in~$R$ to the root by an $x$- and $y$-monotone path, traces can be regarded as generalized RSAs.
In particular, we make the following observation.

\begin{observation}\label{obs_arborescence_in_trace}
Let $R$ be a covering trace for $S$ that contains no boxes, and
let $A_{\tau}$ be a shortest path tree in~$R$ from the root to
all sinks in~$S$.
Then $A_{\tau}$  is an RSA for~$S$.\noproof
\end{observation}

If $\tau$ contains no flips that change bends, the corresponding trace $R$ has no boxes. 
Then, $R$ contains an RSA $A_{\tau}$ with $2|A_{\tau}| \leq \text{cost}(R)$, by Observation~\ref{obs_arborescence_in_trace}.
The next lemma shows that, due to the fact that~$\beta$ is even, there is always a shortest covering trace for $S$ that does not contain any boxes.

\begin{lemma}\label{lem:eliminate_chain_phase}
Let $\tau$ be a flip traversal of~$S$.
Then there exists a covering trace $R$ for $S$ 
such that $R$ does not contain a box and such that $\text{cost}(R) \leq |\tau|$.
\end{lemma}
To prove the lemma, we investigate the structure of minimal covering traces.
There exists at least one trace of cost at most $|\tau|$, namely the trace of $\tau$.
Let $\mathcal{R}_1$ be the set of all covering traces for $S$ that have minimum cost.
Let $\mathcal{R}_2 \subseteq \mathcal{R}_1$ be those covering traces among $\mathcal{R}_1$ that contain the  minimum number of boxes.
If $\mathcal{R}_2$ contains a trace without boxes, we are done, as every covering trace in $\mathcal{R}_2$ fulfills the requirements of Lemma~\ref{lem:eliminate_chain_phase}.
We show that this is actually the case by assuming, for the sake of contradiction, that every covering trace in $\mathcal{R}_2$ contains at least one box.

Let  $R \in \mathcal{R}_2$ and suppose that~$R$ contains a box.
Let $B$ be a \emph{maximal} box in~$R$, i.e., $R$ has no other box whose lower left corner has 
both $x$- and $y$-coordinate at least as large as the lower left corner of~$B$.
In order to prove Lemma~\ref{lem:eliminate_chain_phase}, we need several lemmata on traces of minimum cost.

\begin{figure}
\begin{center}
\includegraphics[width=\textwidth]{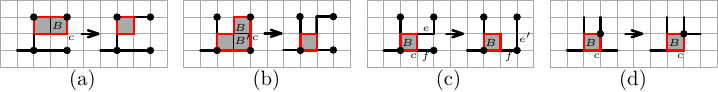}
\end{center}
\caption{
Parts of traces to be modified; the boundary sides are shown in red.
(a) A box that has a corner $c$ with no incident elements can be removed.
(b) Two adjacent boxes that have a shared corner $c$ without any incident elements can be removed.
(c) Replacing a single edge.
(d) Sliding an edge.
}
\label{fig:tree_operations}
\end{figure}

\begin{lemma}\label{lem:no_dangling_corner}
Let $B$ be a maximal box and let $c$ a corner of $B$ that is not the root $(0,0)$.
Then $c$ is incident either to a sink, an edge, or another box.
\end{lemma}
\begin{proof}
Suppose there exists a corner~$c$ for which this is not the case.
Note that such a~$c$ cannot be the lower left corner of~$B$, as there has to be an $x$- and $y$-monotone path to the root.
Hence, we could remove $c$ and $B$ while keeping the sides of~$B$ not incident to~$c$ as edges, if necessary;
see \figurename~\ref{fig:tree_operations}(a).
In the resulting structure, every element still has an $x$- and $y$-monotone path to the root:
If $c$ is the lower right or upper left corner, any path initially passing through $c$ could be rerouted to pass through the corner opposite of~$c$ in~$B$.
If $c$ is the upper right corner of~$B$, no path is passing through~$c$.
Hence, the resulting structure would be a covering trace with smaller cost, contradicting the choice of~$R$.
\qedopt
\end{proof}

\begin{lemma}\label{lem:no_double_box}
Suppose $B$ shares a horizontal side with another box $B'$.
Let $c$ be the right endpoint of the common side.
Then $c$ is incident either to a sink, an edge, or another box.
\end{lemma}
\begin{proof}
Suppose this is not the case.
Then we could remove $B$ and $B'$ from $R$ while keeping the sides not incident to $c$ as edges, if necessary; see \figurename~\ref{fig:tree_operations}(b).
This results in a valid trace that has no higher cost but less boxes than $R$, contradicting the choice of~$R$.
\qedopt
\end{proof}

\begin{lemma}\label{lem:no_vertical_edge}
Let $c$ be the lower right corner of~$B$.
Then $c$ has no incident vertical edge.
\end{lemma}
\begin{proof}
Such an edge would be redundant, since $c$ already has an $x$- and $y$-monotone path to the root that goes through the lower left corner of~$B$.
\qedopt
\end{proof}

\begin{proof}[Proof of Lemma~\ref{lem:eliminate_chain_phase}]
Using the Lemmata~\ref{lem:no_dangling_corner}, \ref{lem:no_double_box}, and \ref{lem:no_vertical_edge}, we derive a contradiction from the choice of~$R$ and the maximal box~$B$.
Note that since $\beta$ is even, all sinks in~$S$ have even $x$- and $y$-coordinates.
We distinguish two cases.

\noindent\textbf{Case 1}.
There exists a maximal box $B$ whose top right corner $c'$ does 
not have both coordinates even.
Suppose that the $x$-coordinate of~$c'$ is odd. (Otherwise, 
mirror the plane at the line $x=y$ to swap the $x$- and the $y$-axis.
Note that the property of being a trace is invariant under mirroring 
the plane along the line $x = y$;
in particular, the choice of~$B$ in~$R$ as a maximal box remains valid)
By Lemma~\ref{lem:no_dangling_corner},
there is at least one edge incident to the top right corner of~$B$ 
(it cannot be a box by the choice of~$B$, and it cannot be a sink 
because of the current case). 
Recall the slide operation for an edge in an arborescence. This operation 
can easily be adapted in an analogous way to traces.
If there is a vertical edge $v$ incident to~$c'$, it cannot be incident
to a sink. Thus, we could slide $v$ 
to the right (together with all other 
vertical edges that are above $v$ and on the supporting line of~$v$).
Hence, we may assume that $c'$ is incident to a single horizontal edge~$e$;
see \figurename~\ref{fig:tree_operations}(c).
By Lemma~\ref{lem:no_dangling_corner}, the bottom right corner $c$ 
of~$B$ must be incident to an element.
We know that $c$ cannot be the top right corner of another box 
(Lemma~\ref{lem:no_double_box}), nor can it be incident to a vertical
segment (Lemma~\ref{lem:no_vertical_edge}). Thus, $c$ is incident to
an element~$f$ that is either a horizontal edge or a box with top left corner~$c$.
But then~$e$ could be replaced by a vertical segment $e'$ incident to~$f$,
and afterwards  $B$ could be removed as in the proof of
Lemma~\ref{lem:no_dangling_corner}, contradicting the choice of~$R$.

\noindent\textbf{Case 2}.
The top right corner of each maximal box has even coordinates.
Let $B$ be the rightmost maximal box.
As before, let $c$ be the bottom right corner of~$B$.
The $y$-coordinate of~$c$ is odd; see \figurename~\ref{fig:tree_operations}(d).
By the choice of~$B$, we know that $c$ is not the top left corner of another box: this would imply that there is another maximal box to the right of~$B$.
We may assume that $c$ is not incident to a horizontal edge, as we could slide such an edge up, as in Case 1.
Furthermore, $c$ cannot be incident to a vertical edge (Lemma~\ref{lem:no_vertical_edge}), nor be the top right corner of another box (Lemma~\ref{lem:no_double_box}).
Thus, $B$ violates Lemma~\ref{lem:no_dangling_corner}, and Case~2 also leads to a contradiction.

Thus, the choice of~$R$ forces a contradiction in either case.
Hence, the minimum number of boxes in a minimum covering trace for~$S$ is~$0$.
\end{proof}

\noindent
Now we can finally complete the proof of Theorem~\ref{thm:YRSA<->FlipDist} by giving the second direction of the correspondence.

\begin{lemma}\label{lem:flip_tour->YRTSP}
Let $k \geq 1$ and
let $\sigma$ be a flip sequence on $P$ from $T_1$ to
$T_2$ with $|\sigma| \leq 2 \beta k + (4d-2)N$.
Then there exists an RSA for~$S$ of length at most~$\beta k$.
\end{lemma}
\begin{proof}
Trivially, there always exists an RSA on $S$ of length less than $2\beta nN$, so we may assume that 
$k < 2nN$.  Hence (recall that $\beta = 2N$ and $d = nN$),
\[
2\beta k + (4d-2)N < 2 \cdot 2N \cdot 2nN + 4 nN^2 -2N < 12nN^2 <  (d-1)^2,
\] 
for $n \geq 14$ and positive~$N$.
Thus, $\sigma$ meets the requirements of Lemma~\ref{lem:flip_traversal}, and therefore
we can obtain a flip traversal $\tau$ for~$S$ with $|\tau| \leq 2\beta k + 2N$.
By Lemma~\ref{lem:eliminate_chain_phase} and Observation~\ref{obs_arborescence_in_trace},
we can conclude that there is an RSA $A$ for $S$ 
that has length at most $\beta k + N$. 
By Theorem~\ref{thm_slide}, there is an RSA~$A'$ for $S$ that is not 
longer than $A$ and that lies on 
the Hanan grid for $S$. The length of $A'$ must be a multiple of~$\beta$.
Thus, since $\beta > N$, we get that $A'$ has length at most $\beta k$.
\qedopt
\end{proof}

\section{Conclusion}
In this paper, we showed NP-hardness of determining a shortest flip sequence between two triangulations of a simple polygon.
This complements the recent hardness results for point sets (obtained by reduction from variants of \textsc{Vertex Cover}).
However, while for point sets the problem is hard to approximate as well, our reduction does not rule out the existence of a polynomial-time approximation
scheme (PTAS), since a PTAS is known for the RSA problem~\cite{rsa_ptas}.
When problems that are hard for point sets are restricted to simple polygons, the application of standard techniques---like dynamic programming---often gives polynomial-time algorithms.
This is, for example, the case for the construction of the minimum weight triangulation.
Our result illustrates that determining the flip distance is a different, harder type of problem.
Is there a PTAS for the flip distance between triangulations of a polygon?
Even a constant-factor approximation would be interesting.

For convex polygons (or, equivalently, points in convex position), the complexity of the problem remains unknown.
Our construction heavily relies on the double chain construction, using many reflex vertices.
Does the problem remain hard if we restrict the number of reflex vertices to some constant fraction?

\bibliographystyle{spmpsci}
\bibliography{bibliography}

\begin{thebibliography}{10}
\providecommand{\url}[1]{{#1}}
\providecommand{\urlprefix}{URL }
\expandafter\ifx\csname urlstyle\endcsname\relax
  \providecommand{\doi}[1]{DOI~\discretionary{}{}{}#1}\else
  \providecommand{\doi}{DOI~\discretionary{}{}{}\begingroup
  \urlstyle{rm}\Url}\fi

\bibitem{empty5gon}
Abel, Z., Ballinger, B., Bose, P., Collette, S., Dujmovi{}\'c, V., Hurtado, F.,
  Kominers, S., Langerman, S., P{\'o}r, A., Wood, D.: Every large point set
  contains many collinear points or an empty pentagon.
\newblock Graphs Combin. \textbf{27}, 47--60 (2011)

\bibitem{eurocg_version}
{Aichholzer}, O., {Mulzer}, W., {Pilz}, A.: {Flip Distance Between
  Triangulations of a Simple Polygon is NP-Complete}.
\newblock In: Proc. 29\textsuperscript{th} European Workshop on Computational
  Geometry, pp. 115--118. Braunschweig, Germany (2013)

\bibitem{esa_version}
Aichholzer, O., Mulzer, W., Pilz, A.: Flip distance between triangulations of a
  simple polygon is {NP}-complete.
\newblock In: H.L. Bodlaender, G.F. Italiano (eds.) Algorithms - {ESA} 2013 -
  21st Annual European Symposium, Sophia Antipolis, France, September 2-4,
  2013. Proceedings, \emph{Lecture Notes in Computer Science}, vol. 8125, pp.
  13--24. Springer (2013)

\bibitem{survey}
Bose, P., Hurtado, F.: Flips in planar graphs.
\newblock Comput. Geom. \textbf{42}(1), 60--80 (2009)

\bibitem{canny}
Canny, J.F., Donald, B.R., Ressler, E.K.: A rational rotation method for robust
  geometric algorithms.
\newblock In: Proc. 8\textsuperscript{th} Symposium on Computational Geometry
  (SoCG 1992), pp. 251--260 (1992)

\bibitem{power_of_duality}
Chazelle, B., Guibas, L.J., Lee, D.T.: The power of geometric duality.
\newblock BIT \textbf{25}(1), 76--90 (1985)

\bibitem{tree_similarity}
Culik~II, K., Wood, D.: A note on some tree similarity measures.
\newblock Inf. Process. Lett. \textbf{15}(1), 39--42 (1982)

\bibitem{constructing_arrangements}
Edelsbrunner, H., O'Rourke, J., Seidel, R.: Constructing arrangements of lines
  and hyperplanes with applications.
\newblock SIAM J. Comput. \textbf{15}(2), 341--363 (1986)

\bibitem{eppstein}
Eppstein, D.: Happy endings for flip graphs.
\newblock JoCG \textbf{1}(1), 3--28 (2010)

\bibitem{edge_flipping_distance}
Hanke, S., Ottmann, T., Schuierer, S.: The edge-flipping distance of
  triangulations.
\newblock J.UCS \textbf{2}(8), 570--579 (1996)

\bibitem{hurtado_noy_urrutia}
Hurtado, F., Noy, M., Urrutia, J.: Flipping edges in triangulations.
\newblock Discrete Comput. Geom. \textbf{22}, 333--346 (1999)

\bibitem{husemoeller}
Husem{\"o}ller, D.: Elliptic Curves.
\newblock Graduate Texts in Mathematics. Springer-Verlag, New York, NY, USA
  (2003)

\bibitem{hwang}
Hwang, F., Richards, D., Winter, P.: The {Steiner} Tree Problem.
\newblock Annals of Discrete Mathematics. North-Holland (1992)

\bibitem{flip_distance_fpt}
Kanj, I.A., Xia, G.: Flip distance is in {FPT} time {$O(n+ k \cdot c^k)$}.
\newblock In: 32nd International Symposium on Theoretical Aspects of Computer
  Science, {STACS}, pp. 500--512 (2015)

\bibitem{lawson_connected}
Lawson, C.L.: Transforming triangulations.
\newblock Discrete Math. \textbf{3}(4), 365--372 (1972)

\bibitem{lawson_delaunay}
Lawson, C.L.: Software for {$C^1$} surface interpolation.
\newblock In: J.R. Rice (ed.) Mathematical Software III, pp. 161--194. Academic
  Press, NY (1977)

\bibitem{rsa_ptas}
Lu, B., Ruan, L.: Polynomial time approximation scheme for the rectilinear
  {S}teiner arborescence problem.
\newblock J. Comb. Optim. \textbf{4}(3), 357--363 (2000)

\bibitem{lubiw}
{Lubiw}, A., {Pathak}, V.: Flip distance between two triangulations of a
  point-set is {NP}-complete.
\newblock In: Proc. 24\textsuperscript{th} CCCG, pp. 127--132 (2012)

\bibitem{point_set_hard}
Pilz, A.: Flip distance between triangulations of a planar point set is
  {APX}-hard.
\newblock Comput. Geom. \textbf{47}(5), 589--604 (2014)

\bibitem{rao}
Rao, S.K., Sadayappan, P., Hwang, F.K., Shor, P.W.: The rectilinear {S}teiner
  arborescence problem.
\newblock Algorithmica \textbf{7}, 277--288 (1992)

\bibitem{shi_su}
Shi, W., Su, C.: The rectilinear {S}teiner arborescence problem is
  {NP}-complete.
\newblock In: Proc. 11\textsuperscript{th} SODA, pp. 780--787 (2000)

\bibitem{sleator}
Sleator, D., Tarjan, R., Thurston, W.: Rotation distance, triangulations and
  hyperbolic geometry.
\newblock J. Amer. Math. Soc. \textbf{1}, 647--682 (1988)

\bibitem{trubin}
Trubin, V.: Subclass of the {S}teiner problems on a plane with rectilinear
  metric.
\newblock Cybernetics \textbf{21}, 320--324 (1985)

\bibitem{problemas}
Urrutia, J.: Algunos problemas abiertos.
\newblock In: Proc. IX Encuentros de Geometr{\'i}a Computacional, pp. 13--24
  (2001)

\end{thebibliography}

\newpage
\appendix

\section{A Note on Coordinate Representation}\label{apx_coordinates}
Since it is necessary for the validity of the proof that the input polygon can be represented in size polynomial in the size of the YRSA instance, we give a possible method to embed the polygon with vertices at rational coordinates whose numerator and denominator are polynomial in~$N$.
By an additional perturbation argument we can guarantee integer coordinates whose values are polynomial in~$N$ (which slightly strengthens the result).
We first introduce the general technique used for the embedding, and then give further details on how the sink gadgets are constructed (using methods similar to~\cite{point_set_hard}).
Finally, we explain how the construction can be transformed to integer points in general position.

\subsection{Placing Points on Arcs}

The main idea of the construction is to place all vertices on rational points on circular arcs.
There are two large arcs where we place the vertices of the upper and the lower chain, and smaller arcs on which we place the vertices of the sink gadgets.
All these circular arcs are chosen from \emph{rational circles}, i.e., circles that are defined by three rational points.
Similarly, a \emph{rational line} is a line trough a rational point with rational slope (or, equivalently, a line defined by two rational points).
It is well-known that, if one of the two intersection points of a rational line with a rational circle is a rational point, then the other intersection point is rational as well (see, e.g.,~\cite[p.~5]{husemoeller}).
Hence, given a rational point~$p$ on a rational circle, we can obtain an arbitrary number of rational points on the circle via different rational lines through~$p$.

Let us apply this for one possible way of constructing the double chain~$D$.
The construction is shown in \figurename~\ref{fig_exact_construction}~(left).
We place the $\beta n + 2$ points of the lower chain on the unit circle (with center at the origin).
Let $\ell_i$ be the line through $(-1,0)$ with slope $1 + \frac{i}{\beta n + 2}$.
For $i = 1, \dots, \beta n / 2 + 1$, we get $\beta n / 2 + 1$ rational points on the upper-left quadrant of the unit circle from the intersections with this family of lines.
We can do the analogous construction for points on the upper-right quadrant by choosing lines through $(1, 0)$ with a negative slope $(-1) - \frac{i}{\beta n+2}$.
In this way, we obtain the vertices of the lower chain of~$D$.
For the upper chain, we place points on the unit circle with origin $(0, 3)$ analogously.
Note that line $\ell_{\beta n / 2 + 1}$ passes through $(1, 3)$, so when picking rational points on the lower-right and lower-left quadrant of the second unit circle for the upper chain, the resulting point set is indeed the vertex set of a double chain in which the line through $l_0$ and $u_{\beta n+1}$ is~$\ell_{\beta n / 2 + 1}$.
Finally, note that all slopes used in the construction have numerators and denominators that are polynomial in~$N$.
Hence, this also holds for the coordinates of the vertices of~$D$.
Note that this is, essentially, the parametrization of the unit circle, as discussed in~\cite{canny}.

\begin{figure}
\centering
\includegraphics{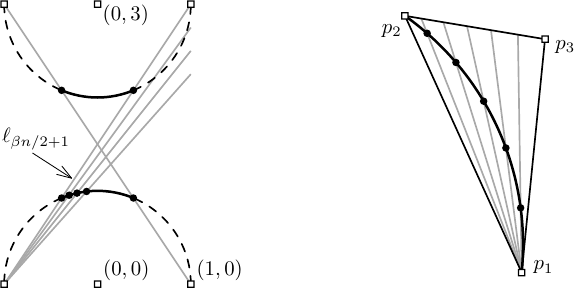}
\caption{Left: Construction of the main double chain~$D$.
Right: Picking points on a circular arc inside a triangle. The line $p_1 p_3$ is tangent to the corresponding circle.}
\label{fig_exact_construction}
\end{figure}

Clearly, this method is not restricted to unit circles.
We now discuss the following main building block for constructing the sink gadgets.
Given three rational points $p_1, p_2, p_3$, we construct a circular arc on a rational circle that starts at $p_1$, ends at $p_2$ and is completely contained inside the triangle $p_1 p_2 p_3$.
Then, we choose an arbitrary number of rational points on that circular arc.
This is illustrated in \figurename~\ref{fig_exact_construction}~(right).
W.l.o.g, let the inner angle of the triangle at $p_1$ be less than or equal to the one at~$p_2$.
Let $Z$ be the circle through $p_1$ and $p_2$ such that the line $p_1 p_3$ is a tangent of~$Z$.
Clearly, $Z$ is well-defined, and the arc between $p_1$ and $p_2$ is inside the triangle.
The circle $Z$ is rational.
(Consider the line that is perpendicular to the line $p_1 p_3$ and passes through~$p_1$.
When mirroring $p_2$ with that line as an axis, the resulting point $p_2'$ is rational and also on~$Z$.)
We can now choose any number of rational points on the circular arc by selecting a family of lines through $p_1$.
To this end, we choose a set of equidistant points on the segment $p_2 p_3$, which, together with $p_1$ define this family of rational lines.
Again, the numerators and the denominators are polynomial in those of $p_1$, $p_2$, and~$p_3$, and the number of points chosen.

\subsection{Constructing Sink Gadgets}
We now construct the sink gadgets.
See \figurename~\ref{fig_coordinates} for an accompanying illustration.
Recall that, since $\beta$ is even, there are no small double chains on neighboring positions on the lower chain.
Hence, for each sink we w.l.o.g.\ can define an orthogonal region within which we can safely draw the small double chain; we call this region the \emph{bin} of the sink (outlined gray in \figurename~\ref{fig_coordinates}).
Consider a sink~$(i,j)$.
The vertical line bounding the left side of its bin passes through the edge $l_{j-1} l_j$ (e.g., at the midpoint of the edge), and the right side of the bin is defined analogously.
(Recall that, since $\beta > 1$, there is no sink at $l_{j-1} l_j$.)
Pick a rational point $p_a$ on the boundary of the bin that is to the left of the directed line $l_j u_{i-1}$ and to the right of the directed line~$l_j u_i$.
Similarly, choose a point $p_b$ that is to the right of the line $l_{j+1} u_{i+1}$ and to the left of the line $l_{j+1} u_i$.
As an additional constraint let $p_a$ be to the left of the line~$p_b u_i$.
Note that such points always exist, and can be easily chosen along the boundary of the bin.
It remains to choose a triangular region with $l_j p_a$ as one side in which we can place the chain of the sink gadget that contains~$l_j$.
For the second chain, the construction is analogous.

For the chain to be visible from $u_i$ but not from $u_{i-1}$, the triangular region has to be to the left of the line $p_a u_i$, and also to the left of the line $l_j u_{i-1}$.
Further, to be visible from all vertices of the other chain, it has to be to the left of the lines $p_a l_{j+1}$ and $p_b l_j$.
Let $x_a$ be the apex of the triangle that is defined by these constraints, and observe that $x_a$ is the intersection of two of the four lines.
We can now add a chain of points on a circular arc inside the triangle $l_j p_a x_a$, as described above.

\begin{figure}
\centering
\includegraphics{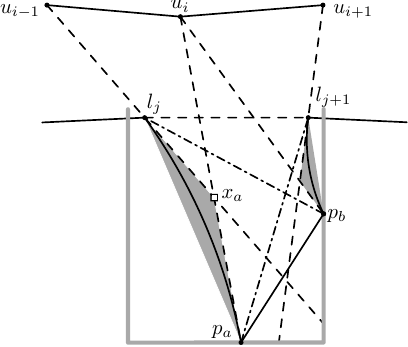}
\caption{Construction of a small double chain for a sink.}
\label{fig_coordinates}
\end{figure}

The coordinates are rational, and since every point can be constructed using only a constant number of other points, the numerator and denominator of each point are polynomial.

\subsection{General Position and Integer Coordinates}
The ways in which a simple polygon can be triangulated is determined by the \emph{order type} of the vertices, i.e., the vector that indicates for
each triple of vertices whether it is oriented clockwise or counterclockwise.
Up to now, we did not care whether the point set is in general position, so there might also be collinear point triples among vertices that are not directly related in the reduction.
By simply multiplying all coordinates by all denominators used, we would obtain integer coordinates with exponential values.
To obtain integer coordinates bounded by a polynomial in the input size and a point set in general position, we can use the following lemma.%
\footnote{The exact time bounds shown in the proof are irrelevant for the NP-hardness reduction (which even requires a different model of computation).
We mention them only as they may be of general interest.}

\begin{lemma}\label{lem_integer_coordinates}
Let~$S$ be a point set with rational coordinates whose numerators and denominators have absolute values of at most~$\xi$.
Then there is a point set~$S'$ with integer coordinates bounded by~$O(\xi^3)$ and a bijection between $S$ and $S'$ such that for every ordered triple of non-collinear points in~$S$, the orientation of the corresponding triple in~$S'$ is the same.
In particular, if $S$ is in general position, then $S$ and $S'$ have the same order type.
Further, $S'$ can be constructed in $O(|S|^2)$ time.
\end{lemma}
\begin{proof}
Consider the set~$L$ of lines that are defined by all pairs of points of~$S$.
Choose $\ell = q_1 q_2 \in L$ and $p \in S \setminus \ell$ such that the horizontal distance $v$ between $p$ and $\ell$ is minimal among all such distances (which is non-zero as $p$ is not on $\ell$). %
Then $v$ is rational, with numerator and denominator in~$O(\xi^2)$.
Further, our choice required $v > 0$.
When multiplying all $x$-coordinates by $2/v$, this distance is at least~$2$.
The basic idea is to round the $x$-coordinates.
The crucial observation is that $p$ has a $y$-coordinate that is between the ones of $q_1$ and $q_2$, as otherwise one of $q_1$ and $q_2$, say, $q_1$, would be horizontally closer to the line through $p$ and $q_2$.
For an ordered triple of points to change its orientation (from, say, clockwise to counterclockwise), the horizontal distance between the point whose $y$-coordinate is between those of the other two points would have to be reduced by more than~$2$.
We can therefore safely round the $x$-coordinates, which, in the worst case, reduces the horizontal distance between $p$ and $\ell$ by at most~$1$.
Hence, for every non-collinear ordered triple of points in~$S$, the orientation of the corresponding triple in the resulting point set is the same.
We repeat the process analogously for the $y$-coordinates, obtaining~$S'$.

The horizontal or vertical distance $v$ can easily be found by checking all triples of points.
We can improve this cubic time bound by considering the dual line arrangement~$\mathcal{A}$ of~$S$ (in which a point $p = (x_p, y_p)$ corresponds to the dual line $p^* : y = x \cdot x_p + y_p$).
The dual arrangement can be constructed in quadratic time~\cite{power_of_duality,constructing_arrangements}.
The shortest vertical distance in the primal corresponds to the shortest vertical distance of a vertex and a side of a triangle defined by three dual lines.
Clearly, the shortest distance can only occur inside a triangle that is not intersected by another line.%
\footnote{Actually, any dual transform will do.
When thinking of the rounding process as a continuous transformation, a change of the order type would involve a collapsing triangular cell of the dual arrangement, indicating a ``close'' point triple.}
Hence, we only need to test the $O(|S|^2)$ triangular cells of~$\mathcal{A}$.
\end{proof}

Hence, if we construct the vertices with rational coordinates such that the vertices are in general position, we can apply Lemma~\ref{lem_integer_coordinates} to have all vertices on the integer grid in general position.

General position can easily be obtained by applying a simple technique used in~\cite[Appendix~A]{point_set_hard}.
Observe first that the vertices of $D$ are in general position.
We take special care when placing the $d-1$ points of each chain of a sink gadget to not produce collinear points.
Note that the final polygon $P$ will have $|P|= 2(\beta n + 2) + 2N (d-1)$ vertices.
Instead of $d-1$ points, we choose $2\binom{|P|}{2} + d - 1$ \emph{candidate points} on the circular arc for the chain.
Consider any line through two already placed points.
This line intersects the circular arc in at most two points, so there are at most two candidate points that may not be points of the double chain because of that line.
As there are less than $\binom{|P|}{2}$ such lines, there are always enough candidate points left for selecting the $d-1$ points for the chain among them.
Thus, the vertices we obtain are in general position.

\end{document}